\documentclass[aps, pra, reprint, superscriptaddress, onecolumn, notitlepage, tightenlines, 12pt]{revtex4-1}

\pdfoutput=1

\usepackage{header}

\begin{document}

\bibliographystyle{acm}

\title{On the Parameterised Complexity of Induced Multipartite Graph Parameters}

\author{Ryan L. Mann}
\email{mail@ryanmann.org}
\homepage{http://www.ryanmann.org}
\affiliation{School of Mathematics, University of Bristol, Bristol, BS8 1UG, United Kingdom}
\affiliation{Centre for Quantum Computation and Communication Technology, \\ Centre for Quantum Software and Information, \\ Faculty of Engineering \& Information Technology, University of Technology Sydney, NSW 2007, Australia}

\author{Luke Mathieson}
\email{luke.mathieson@uts.edu.au}
\affiliation{School of Computer Science, University of Technology Sydney, NSW 2007, Australia}

\author{Catherine Greenhill}
\email{c.greenhill@unsw.edu.au}
\affiliation{School of Mathematics and Statistics, UNSW Sydney, NSW 2052, Australia}

\begin{abstract}
    We introduce a family of graph parameters, called \emph{induced multipartite graph parameters}, and study their computational complexity. First, we consider the following decision problem: an instance is an induced multipartite graph parameter $p$ and a given graph $G$, and for natural numbers $k\geq2$ and $\ell$, we must decide whether the maximum value of $p$ over all induced $k$-partite subgraphs of $G$ is at most $\ell$. We prove that this problem is \mbox{\textsc{W[1]}-hard}. Next, we consider a variant of this problem, where we must decide whether the given graph $G$ contains a sufficiently large induced $k$-partite subgraph $H$ such that $p(H)\leq\ell$. We show that for certain parameters this problem is \mbox{para-\textsc{NP}-hard}, while for others it is fixed-parameter tractable.
\end{abstract}

\maketitle

\section{Introduction}
\label{section:Introduction}

Structural parameters play an important role in characterising the computational complexity of graph problems. A canonical example of this is Courcelle's theorem~\cite{courcelle1990monadic, courcelle1992monadic}, which states that every graph problem definable in monadic second-order logic can be solved in fixed-parameter tractable time with respect to the treewidth. Since determining the treewidth of a graph is itself fixed-parameter tractable~\cite{bodlaender1996linear}, this provides an algorithmic template for solving many graph problems.

In this paper we introduce a family of graph parameters called \emph{induced multipartite graph parameters}. Informally, an induced multipartite graph parameter is a graph parameter that (1) is non-increasing for subgraphs of complete $k$-partite graphs, and (2) is an efficiently computable injective function over the set of complete $k$-partite graphs with $k\geq2$. This family includes a number of well-studied graph parameters, including vertex and edge connectivity, chromatic index, treewidth, and pathwidth.

We firstly consider the problem of deciding whether the maximum value of an induced multipartite graph parameter, over all induced $k$-partite subgraphs of the input graph, is at most some constant $\ell$. By establishing a reduction from the maximum stable set problem, we prove that this problem is \mbox{\textsc{W[1]}-hard} when parameterised by $k$ and $\ell$. Secondly, we consider the problem of deciding whether a graph has a sufficiently large induced $k$-partite subgraph $H$ such that the value of a given induced multipartite graph parameter on $H$ is at most a (specified) constant $\ell$. We show that the complexity of this problem can vary, depending on the specific graph parameter. For example, this problem, when parameterised by $k$, $\ell$, and the number of deleted vertices $m$, is \mbox{para-\textsc{NP}-hard} for minimum degree, maximum degree, vertex connectivity, edge connectivity, and chromatic index; and is fixed-parameter tractable for number of vertices and number of edges.

A special case of the former problem is the \emph{bipartite pathwidth}, which is defined to be the maximum pathwidth of any induced bipartite subgraph. Dyer, Greenhill, and M\"uller~\cite{dyer2021counting} introduced this concept and showed that, for graphs of bounded bipartite pathwidth, the Markov chain known as the \emph{Glauber dynamics for stable sets} is rapidly mixing. As a consequence, they obtained an efficient randomised approximation scheme for counting stable sets on graphs of bounded bipartite pathwidth, and an efficient algorithm for sampling stable sets of a given size. However, \mbox{Theorem~\ref{theorem:InducedKPartiteSubgraphParameterHardness}} implies that it is a computationally hard problem to decide whether a given graph has bounded bipartite pathwidth.

This paper is structured as follows. In \mbox{Section~\ref{section:InducedMultipartiteGraphsParameters}}, we introduce the notion of induced multipartite graph parameters and provide several examples. Then, in \mbox{Section~\ref{section:HardnessResults}}, we consider the problem of deciding, for a given graph $G$ and induced multipartite graph parameter $p$, whether the maximum value of $p$ over induced $k$-partite subgraphs of $G$ is at most a specified constant. We prove \mbox{Theorem~\ref{theorem:InducedKPartiteSubgraphParameterHardness}} which states that this problem is \mbox{\textsc{W[1]}-hard}. We then consider the problem of deciding whether, for a given graph $G$ and induced multipartite graph parameter $p$, there exists a sufficiently large induced multipartite subgraph $H$ of $G$ such that $p(H)$ is at most a specified constant. We prove that this problem is \mbox{para-\textsc{NP}-hard} for minimum degree, maximum degree, vertex connectivity, edge connectivity, and chromatic index, while in \mbox{Section~\ref{section:FixedParameterTractability}}, we prove that this problem is fixed-parameter tractable for number of vertices and number of edges. Finally, we conclude in \mbox{Section~\ref{section:ConclusionAndFurtherWork}} with some remarks and open problems. 

\section{Induced Multipartite Graph Parameters}
\label{section:InducedMultipartiteGraphsParameters}

All graphs in this paper are finite and simple. A graph parameter is a function $p$ that assigns to every graph $G$ a real number $p(G)$ which is invariant under vertex relabelling. We restrict our attention to graph parameters that take non-negative integer values, but extensions to rational values are also possible.

Write $K_{n|k}$ for the complete $k$-partite graph with $n$ vertices in each part of the vertex partition, where $n$ and $k$ are positive integers.

\begin{definition}[Induced multipartite graph parameter]
    A graph parameter $p$ is an \emph{induced multipartite graph parameter} if the following properties are satisfied.
    \begin{enumerate}
        \item[\emph{\textbf{P1.}}] If $H$ is a subgraph of $K_{n|k}$, then \mbox{$p(H) \leq p(K_{n|k})$}.
        \item[\emph{\textbf{P2.}}] For each $k\geq2$ there is an efficiently computable injective function \mbox{$f_k:\mathbb{Z}^+\to\mathbb{N}$} such that \mbox{$p(K_{n|k})=f_k(n)$}. 
    \end{enumerate}
    For a given induced multipartite graph parameter $p$, let $p(G,k)$ be the maximum value of $p(H)$ over all induced $k$-partite subgraphs $H$ of $G$.
\end{definition}

Note that when \textbf{P1} holds, the function $f_k$ in \textbf{P2} is injective if and only if it is strictly increasing.

A \emph{monotone} graph property is a graph property which is inherited by subgraphs~\cite{alon2008every}. Similarly, we say that a graph parameter $p$ is \emph{monotone} if \mbox{$p(H) \leq p(G)$} for any graph $G$ and any subgraph $H$ of $G$. Note that any monotone graph parameter will satisfy \textbf{P1}. We shall now proceed to show that many common graph parameters are induced multipartite graph parameters.

Write \mbox{$\abs{G}=\abs{V(G)}$} for the number of vertices in $G$ (the \emph{order} of $G$), \mbox{$\|G\|=\abs{E(G)}$} for the number of edges in $G$ (the \emph{size} of $G$), $\delta(G)$ for the minimum degree of $G$ and $\Delta(G)$ for the maximum degree of $G$.

A graph $G$ is $k$-\emph{connected} if $\abs{G}>k$ and $G$ cannot be disconnected by deleting fewer than $k$ vertices. Write $\kappa(G)$ for the \emph{vertex connectivity} of $G$, that is, the largest $k$ such that $G$ is $k$-connected. The \emph{edge connectivity} $\lambda(G)$ is minimum size of a set of edges whose deletion disconnects $G$.

The \emph{chromatic index} $\chi'(G)$ of a graph $G$ is the smallest positive integer $r$ such that the edges of $G$ may be coloured with $r$ colours with no two incident edges receiving the same colour.

The \emph{treewidth} of a graph is a measure of how close the graph is to a tree~\cite{robertson1984graph}. To define the treewidth, we must first define the \emph{tree decomposition} of a graph.
\begin{definition}[Tree decomposition]
    A tree decomposition of a graph $G$ is a tree $T$ and a collection $\mathcal{X}$ of subsets of $V(G)$ indexed by the vertices of $T$, such that the following properties are satisfied.
    \begin{itemize}
        \item For all $v \in V(G)$ there is least one $B\in\mathcal{X}$ with $v \in B$.
        \item For all $v \in V(G)$, the elements of $\mathcal{X}$ that contain $v$ form a subtree of $T$.
        \item For every edge $e \in E(G)$, there is at least one $B \in \mathcal{X}$ with $e \subseteq B$.
    \end{itemize}
    Elements of $\mathcal{X}$ are called \emph{bags}. The \emph{width} of a tree decomposition is \mbox{$\max_{B\in\mathcal{X}}\abs{B}-1$} (the size of the largest bag, minus one). The \emph{treewidth} $\operatorname{tw}(G)$ of a graph $G$ is the minimum width over all possible tree decompositions of $G$.
\end{definition}

The \emph{pathwidth} $\operatorname{pw}(G)$ of a graph $G$ is the minimum width over all possible tree decompositions of $G$ whose tree is a path~\cite{robertson1983graph}. Such a tree decomposition is known as a \emph{path decomposition}.

\begin{lemma}
    Each of the following is an induced multipartite graph parameter: number of vertices, number of edges, minimum degree, maximum degree, vertex connectivity, edge connectivity, chromatic index, treewidth, and pathwidth. 
\end{lemma}

\begin{proof}
    It is immediate that \textbf{P1} holds for number of vertices, number of edges, minimum degree and maximum degree. Now 
    \begin{equation}
        \abs{K_{n|k}} = kn, \quad \|K_{n|k}\| = \binom{k}{2}n^2, \quad \delta(K_{n|k}) = (k-1)n, \quad \Delta(K_{n|k}) = (k-1)n. \notag
    \end{equation}
    For fixed $k$, the corresponding functions of $n$ are all efficiently computable and injective. It follows that \textbf{P2} holds for number of vertices, number of edges, minimum degree and maximum degree. 
    
    If $H$ is a subgraph of $K_{n|k}$, then any proper edge-colouring of $K_{n|k}$ can be restricted to give a proper edge-colouring of $H$. It follows that the chromatic index satisfies \textbf{P1}. Furthermore, it follows from the main theorem of Ref.~\cite{hoffman1992chromatic} that
    \begin{equation}
        \chi'(K_{n|k}) = \begin{cases} 
            (k-1)n & \text{if $kn$ even}, \\
            (k-1)n+1 & \text{if $kn$ odd}. \notag
        \end{cases}
    \end{equation}
    It follows that \textbf{P2} holds for the chromatic index, since $k\geq2$.
    
    Next, recall that for any graph $H$ the vertex and edge connectivity satisfy
    \begin{equation}
        \kappa(H)\leq\lambda(H)\leq\delta(H). \notag
    \end{equation}
    If $H$ is a subgraph of $K_{n|k}$, then \mbox{$\delta(H)\leq(k-1)n$}, and \textbf{P1} and \textbf{P2} follow after observing that
    \begin{equation}
        \kappa(K_{n|k}) = (k-1)n, \quad \lambda(K_{n|k}) = (k-1)n. \notag
    \end{equation}
    
    Now consider treewidth. Let $H$ be a subgraph of $G$ and let $(B_i)_{i=1}^d$ be a tree decomposition of $G$ with treewidth $\operatorname{tw}(G)$. Then \mbox{$(B_i \cap V(H))_{i=1}^d$} is a tree decomposition of $H$ with width at most $\operatorname{tw}(G)$. Hence \mbox{$\operatorname{tw}(H)\leq\operatorname{tw}(G)$}, so treewidth is a monotone parameter. In particular, treewidth satisfies property \textbf{P1}. The same argument holds for pathwidth by restricting to tree decompositions such that the tree is a path. Fijav\v{z} and Wood~\cite[Lemma 8.3]{fijavvz2010graph} proved that when $k\geq2$,
    \begin{equation}
        \operatorname{tw}(K_{n|k}) = \operatorname{pw}(K_{n|k}) = (k-1)n. \notag
    \end{equation}
    Hence both treewidth and pathwidth satisfy property \textbf{P2} with \mbox{$f_k(n)=(k-1)n$}, completing the proof.
\end{proof}

We remark that several parameters listed in the above lemma are monotone (in the sense that we defined earlier: that is, they are monotonically non-increasing with respect to taking subgraphs), while minimum degree, vertex connectivity, and edge connectivity are induced multipartite graph parameters but are not monotone. Some other natural graph invariants are not induced multipartite graph parameters, including the chromatic number $\chi(G)$, clique number $\omega(G)$, and domination number $\gamma(G)$. These parameters satisfy \textbf{P1} but do not satisfy \textbf{P2}, since
\begin{equation}
    \chi(K_{n|k}) = k, \quad \omega(K_{n|k}) = k, \quad \gamma(K_{n|k}) = \min\{2,n\} \notag
\end{equation}
and hence the corresponding function \mbox{$f_k(n)=k$} is not injective.

\section{Hardness Results}
\label{section:HardnessResults}

In this section we shall present two hardness results concerning induced multipartite graph parameters. First, we consider the problem of deciding whether the maximum value of a particular induced multipartite graph parameter $p$ over induced $k$-partite subgraphs is at most a specified constant.

\begin{flushleft}
\begin{tabular}{ll}
    \multicolumn{2}{l}{\textsc{Induced $k$-Partite Subgraph Parameter}.} \\
    \textit{Instance:} & \quad A graph $G$. \\
    \textit{Parameters:} & \quad Natural numbers $k\geq2$ and $\ell$. \\
    \textit{Problem:} & \quad Decide whether \mbox{$p(G,k)$} is at most $\ell$.
\end{tabular}
\end{flushleft}

We shall prove that this problem is \mbox{\textsc{W[1]}-hard}. To achieve this, we show a reduction from the maximum stable set problem, which is known to be \mbox{\textsc{W[1]}-complete}~\cite{downey1999parameterized}. A \emph{stable set} (also called an independent set) in a graph $G$ is a set of vertices with no edges between them. The \emph{independence number} $\alpha(G)$ of a graph $G$ is the order of a largest stable set in $G$. The maximum stable set problem is defined as follows.

\begin{flushleft}
\begin{tabular}{ll}
    \multicolumn{2}{l}{\textsc{Maximum Stable Set}.} \\
    \textit{Instance:} & \quad A graph $G$. \\
    \textit{Parameter:} & \quad A natural number $m$. \\
    \textit{Problem:} & \quad Decide whether $\alpha(G)$ is at most $m$.
\end{tabular}
\end{flushleft}

We now state and prove our main hardness result.
\begin{theorem}
    \label{theorem:InducedKPartiteSubgraphParameterHardness}
    \textsc{Induced $k$-Partite Subgraph Parameter} is \mbox{\emph{\textsc{W[1]}-hard}} for any induced multipartite graph parameter.
\end{theorem}

\begin{proof}
    Let $G$ be an instance of \textsc{Maximum Stable Set} with parameter $m$. For any integer $k\geq 2$, form the graph $G_k$ by taking $k$ disjoint copies of $G$ and adding an edge from every vertex in each copy of $G$ to every vertex in every other copy of $G$. (That is, \mbox{$G_k = K_k \cdot G$} is the lexicographic product of $K_k$ and $G$, see for example Ref.~\cite{hammack2011handbook}.) Let $p$ be any induced multipartite graph parameter and let $f_k$ be the function from \textbf{P2} for this parameter, that is, such that \mbox{$p(K_{n|k})=f_k(n)$} for all positive integers $n$.
    
    First we claim that \mbox{$p(G_k,k)=f_k(\alpha(G))$}, where, recall, $p(G_k,k)$ denotes the maximum of $p(H)$ over all induced $k$-partite subgraphs of $G_k$. To see this, let $H_0$ be an induced subgraph of $G_k$ formed by taking a maximum stable set in each copy of $G$. Then $H_0$ is an induced subgraph of $G_k$ which is a complete $k$-partite graph with $\alpha(G)$ vertices in each part of the vertex partition. Hence \mbox{$p(H_0)=f_k(\alpha(G))$}. Furthermore, if $H$ is any induced $k$-partite subgraph of $G_k$, then $H$ is isomorphic to a subgraph of $H_0$. Then it follows from \textbf{P1} that \mbox{$p(H) \leq p(H_0)$}, which proves the claim.
    
    Next, note that \textbf{P1} and \textbf{P2} imply that $f_k$ is a strictly increasing function of $n$. It follows that \mbox{$\alpha(G) \leq m$} if and only if \mbox{$p(G_k,k) \leq f_k(m)$}. Therefore we can decide whether $\alpha(G) \leq m$ by constructing $G_k$ (in polynomial time), efficiently computing $f_k(m)$ and testing whether \mbox{$p(G_k,k) \leq f_k(m)$}. Hence \textsc{Induced $k$-Partite Subgraph Parameter} is at least as hard as \textsc{Maximal Stable Set}, which is \mbox{\textsc{W[1]}-complete}.
\end{proof}

\mbox{Theorem~\ref{theorem:InducedKPartiteSubgraphParameterHardness}} provides a hardness result for a wide range of induced multipartite graph parameters. In particular, it implies that deciding whether the induced bipartite pathwidth of a graph is at most a constant is \mbox{\textsc{W[1]}-hard}. This answers a question raised by Dyer, Greenhill, and M\"uller~\cite{dyer2021counting}. When the function $f_k$ is linear, the reduction of \mbox{Theorem~\ref{theorem:InducedKPartiteSubgraphParameterHardness}} also provides us with hardness of approximability results.

\begin{corollary}
    \textsc{Induced $k$-Partite Subgraph Parameter} is \mbox{\emph{Poly-\textsc{APX}-hard}} for any induced multipartite graph parameter $p$ such that $f_k$ is linear.
\end{corollary}

\begin{proof}  
    The reduction given in \mbox{Theorem~\ref{theorem:InducedKPartiteSubgraphParameterHardness}} is approximation-preserving when $f_k$ is linear, and so \mbox{Poly-\textsc{APX}-hardness} is
    inherited from the \textsc{Maximal Stable Set}
    problem~\cite{bazgan2005completeness}.
\end{proof}

Recall that $\abs{G}$ denotes the order of $G$, which is the number of vertices in $G$. We now consider the following problem.

\begin{flushleft}
\begin{tabular}{ll}
    \multicolumn{2}{l}{\textsc{Large Induced $k$-Partite Subgraph Parameter.}} \\
    \textit{Instance:} & \quad A graph $G$. \\
    \textit{Parameter:} & \quad Natural numbers $k\geq2$, $\ell$, and $m$. \\
    \textit{Problem:} & \quad Decide whether there is an induced $k$-partite subgraph $H$ of $G$ \\ & \quad such that $p(H)\leq\ell$ and \mbox{$\abs{H}\geq\abs{G}-m$}.
\end{tabular}
\end{flushleft}

\mbox{In Theorem~\ref{theorem:LargeInducedKPartiteHardness}} we show that this problem is \mbox{para-\textsc{NP}-hard} for maximum degree, minimum degree, vertex connectivity, edge connectivity, and chromatic index, while in \mbox{Section~\ref{section:FixedParameterTractability}}, we show that the problem is fixed-parameter tractable for the number of vertices and number of edges.

Recall that a problem is \mbox{para-\textsc{NP}-hard} if it is \mbox{\textsc{NP}-hard} for some constant value of the parameter $k$. We shall argue by reduction involving the following problem.

\begin{flushleft}
\begin{tabular}{ll}
    \multicolumn{2}{l}{\textsc{Tripartite Maximum Degree $4$}.} \\
    \textit{Instance:} & \quad A graph $G$ with maximum degree $4$. \\
    \textit{Problem:} & \quad Decide whether $G$ is tripartite.
\end{tabular}
\end{flushleft}

The problem \textsc{Tripartite Maximum Degree $4$} is \mbox{\textsc{NP}-complete}, see Ref.~\cite{garey1974some}. Note that a graph $G$ is tripartite if and only if $\chi(G)\leq3$ (that is, the vertices of $G$ can be coloured with $3$ colours so that adjacent vertices receive different colours).

\begin{theorem}
    \label{theorem:LargeInducedKPartiteHardness}
    \textsc{Large Induced $k$-Partite Subgraph Parameter} is \mbox{\emph{para-\textsc{NP}-hard}} for the following induced multipartite graph parameters: maximum degree, minimum degree, vertex connectivity, edge connectivity, and chromatic index.
\end{theorem}

\begin{proof}
    We show that \textsc{Tripartite Maximum Degree $4$} can be reduced to \textsc{Large Induced $k$-partite Subgraph Parameter} for any of these problems. Let $G$ be a graph with maximum degree $4$ that is an instance of \textsc{Tripartite Maximum Degree $4$}. We shall give $G$ as input to \textsc{Large Induced $k$-Partite Subgraph Parameter} with $k=3$ and $m=0$, and with the constant $\ell$ defined so that $p(G)\leq\ell$. Clearly we can take $\ell=4$ when $p$ is maximum degree, minimum degree, vertex connectivity or edge connectivity, Vizing's theorem~\cite{vizing1964estimate} guarantees that we can set $\ell=5$ when $p$ is the chromatic index. Therefore $G$ is tripartite if and only if \textsc{Large Induced $k$-Partite Subgraph Parameter} outputs \textsc{yes} for $G$ with parameters $k=3$, $m=0$ and with $\ell$ as defined above. This completes the reduction and proves that \textsc{Large Induced $k$-Partite Subgraph Parameter} is \mbox{para-\textsc{NP}-hard} for each of the problems listed in the theorem statement.
\end{proof}

We note that above proof works for any induced multipartite graph parameter that takes bounded values on $4$-regular graphs.

\section{Fixed-Parameter Tractability}
\label{section:FixedParameterTractability}

In this section we shall present two fixed-parameter tractability results concerning induced multipartite graph parameters. Specifically, we show that \textsc{Large Induced $k$-Partite Subgraph Parameter} is fixed-parameter tractable for number of vertices, and number of edges.

\begin{theorem}
    \textsc{Large Induced $k$-Partite Subgraph Parameter} is fixed-parameter tractable for number of vertices.
\end{theorem}

\begin{proof}
    Let $G$ be the input graph and let $k$, $\ell$, and $m$ be natural numbers with $k\geq2$. If \mbox{$\abs{G}>\ell+m$}, then we output \textsc{no}. Otherwise, consider all possible subsets $S$ of at most $m$ vertices of $G$. There are at most \mbox{$\sum_{i=0}^m\binom{\ell+m}{i}\leq(\ell+m+1)^m$} choices for $S$. For each subset $S$, test whether $G{\setminus}S$ is $k$-partite and satisfies \mbox{$\abs{G{\setminus}S}\leq\ell$}. This can be decided in time \mbox{$O\left(2^{(\ell+m)}(\ell+m)\right)$} using an algorithm of Bj\"orklund et al.~\cite{bjorklund2009set}. If any of these graphs $G{\setminus}S$ are $k$-partite and satisfy \mbox{$\abs{G{\setminus}S}\leq\ell$}, then we output \textsc{yes}, since this \mbox{$H=G{\setminus}S$} is an induced $k$-partite subgraph with \mbox{$\abs{G}-m\leq\abs{H}\leq\ell$}. Otherwise, output \textsc{no}. This gives a fixed-parameter tractable algorithm for \textsc{Large Induced $k$-Partite Subgraph Parameter} for number of vertices.
\end{proof}

\begin{theorem}
    \textsc{Large Induced $k$-Partite Subgraph Parameter} is fixed-parameter tractable for number of edges.
\end{theorem}

\begin{proof}
    Let $G$ be the input graph and $k$, $\ell$, and $m$ natural numbers with $k\geq2$. First count the number of vertices in $G$ with degree greater than $\ell+m$, and let this number be $s$. If $s>m$, then we output \textsc{no} since, regardless of which $m$ vertices we delete from $G$, the resulting graph will have more than $\ell$ edges. Otherwise, delete all $s$ vertices with degree greater than $\ell+m$. This gives an induced subgraph $\widehat{G}$ of $G$ with no vertex of degree greater than $\ell+m$. Define \mbox{$t=(\ell+m)(m-s)+\ell$} and recall that $\|\widehat{G}\|$ denotes the number of edges in $\widehat{G}$. If $\|\widehat{G}\|>t$, then we output \textsc{no}, since deleting $m-s$ vertices from $\widehat{G}$ can reduce the number of edges by at most \mbox{$(\ell+m)(m-s)$}. Otherwise, there are at most $2t$ non-isolated vertices. Ignore any isolated vertices since they can be assigned to arbitrary parts within a $k$-partite subgraph, and deleting them does not reduce the number of edges.
    
    Now consider all possible subsets $S$ of vertices that are non-isolated in $\widehat{G}$, such that \mbox{$\abs{S} \leq m-s$}. There are at most \mbox{$\sum_{i=0}^{m-s}\binom{2t}{i}\leq(2t+1)^{m-s}$}. For each such subset $S$, test whether the graph $\widehat{G}{\setminus}S$ is $k$-partite and satisfies \mbox{$\|\widehat{G}{\setminus}S\|\leq\ell$}. This can be decided in time \mbox{$O\left(2^{2t}t\right)$} using the algorithm of Bj\"orklund et al.~\cite{bjorklund2009set}. If any of these subgraphs $\widehat{G}{\setminus}S$ are $k$-partite and satisfy \mbox{$\|\widehat{G}{\setminus}S\|\leq\ell$}, then we output \textsc{yes}, since \mbox{$H=G{\setminus}S$} is an induced $k$-partite subgraph with \mbox{$\abs{H}\geq\abs{G}-m$} and $\|H\|\leq\ell$. Otherwise, we output \textsc{no}. This gives a fixed-parameter tractable algorithm for \textsc{Large Induced $k$-Partite Subgraph Parameter} when $p$ is number of edges.
\end{proof}

\section{Conclusion and Further Work}
\label{section:ConclusionAndFurtherWork}

We have classified the parameterised complexity of induced multipartite graph parameters. Specifically, we have shown that \textsc{Induced $k$-Partite Subgraph Parameter} is \mbox{\textsc{W[1]}-hard} in general and that \textsc{Large Induced $k$-Partite Subgraph Parameter} is \mbox{para-\textsc{NP}-hard} for minimum degree, maximum degree, vertex connectivity, edge connectivity, and chromatic index; and is fixed-parameter tractable for number of vertices and number of edges.

Several open problems remain, including identifying further induced multipartite graph parameters $p$ and determining whether \textsc{Large Induced $k$-Partite Subgraph Parameter} for $p$ is \mbox{para-\textsc{NP}-hard} or fixed-parameter tractable. Furthermore, it would be interesting to show that the problem \textsc{Induced $k$-Partite Subgraph Parameter} is \mbox{\textsc{W[1]}-complete} for certain parameters.

\section*{Acknowledgements}

We thank Michael Bremner for helpful discussions. RLM was supported by the QuantERA ERA-NET Cofund in Quantum Technologies implemented within the European Union's Horizon 2020 Programme (QuantAlgo project), EPSRC grants EP/L021005/1 and EP/R043957/1, and the ARC Centre of Excellence for Quantum Computation and Communication Technology (CQC2T), project number CE170100012. CG was supported by the Australian Research Council Discovery Project DP190100977. No new data were created during this study.

\bibliography{bibliography}

\end{document}